\documentclass[12pt]{article}

\usepackage{amsmath, amssymb, amsfonts, amsbsy, amsthm, latexsym, color}
\usepackage{graphicx}
\usepackage{enumerate}
\usepackage{authblk}
\usepackage{hyperref}
\everymath{\displaystyle}
\newtheorem{theorem}{Theorem}[section]
\newtheorem{proposition}[theorem]{Proposition}
\newtheorem{lemma}[theorem]{Lemma}
\newtheorem{corollary}[theorem]{Corollary}

\theoremstyle{definition}
\newtheorem{definition}[theorem]{Definition}
\newtheorem{example}[theorem]{Example}

\newtheorem{remark}[theorem]{Remark}

\newcommand{\R}{\mathbb{R}}

\newcommand{\E}{\mathbb{E}}

\newcommand{\Tc}{\mathcal{T}}

\newcommand{\supp}{\mathrm{supp}}
\newcommand{\spn}{\mathrm{span}}

\newcommand{\inc}{\mathrm{inc}}

\newcommand{\rank}{\text{rank}}

\newcommand{\sign}{{\rm{sign}}}
\newcommand{\tr}{{\rm{Tr}}}
\newcommand{\vecc}{{\rm{vec}}}
\newcommand{\ran}{{\rm{Ran}}}

\newcommand{\mi}{{\ell_{\infty}}} 

\DeclareMathOperator*{\argmin}{argmin}

\title{The Masked Matrix Separation Problem: A First Analysis}

\author[1]{Xuemei Chen}
\author[2]{Rongrong Wang}
\affil[1]{Department of Mathematics and Statistics, University of North Carolina Wilmington}
\affil[2]{CMSE and the Department of Mathematics, Michigan State University}

\begin{document}

\date{\today}

\maketitle

\begin{abstract}
Given a known matrix that is the sum of a low rank matrix and a masked sparse matrix, we wish to recover both the low rank component and the sparse component. The sparse matrix is masked in the sense that a linear transformation has been applied on its left. We propose a convex optimization problem to recover the low rank and sparse matrices, which generalizes the robust PCA framework. We provide incoherence conditions for the success of the proposed convex optimizaiton problem, adapting to the masked setting. The ``mask'' matrix can be quite general as long as a so-called restricted infinity norm condition is satisfied. Further analysis on the incoherence condition is provided and we conclude with promising numerical experiments.
\end{abstract}


\section{Introduction}

The problem of separating a known matrix into its low rank and sparse components has found many applications in recommender systems, multimedia processing, statistical model selection, etc.~\cite{C11, CLMW11}. In this work we are interested in a more general setting, what we call the masked matrix separation problem. Specifically, we would like to recover a low rank, or approximately low rank matrix $L_0\in\R^{m\times n}$, and a sparse or approximately sparse matrix $S_0\in\R^{p\times n}$ from 
\begin{equation}\label{equ:r}
M_0 = L_0+HS_0.
\end{equation} 
In the model \eqref{equ:r}, the matrix $H\in\R^{m\times p}$ is known, and is  referred as the \emph{mask}.

The $\ell_1$ norm and the nuclear norm (sum of singular values) $\|\cdot\|_*$ are  natural convex relaxations of the sparsity of $S$ and low rank of $L$ respectively, so we propose the following convex optimization program with a suitable $\gamma$.
\begin{equation}\label{equ:p}
(\hat S, \hat L)=\argmin_{S, L}\gamma\|S\|_1+\|L\|_*, \quad\text{subject to }L+HS=M_0.
\end{equation}


If $H$ is the identity matrix, this reduces to the classical matrix separation problem \cite{C11, CLMW11}.
If $H$ is invertible, then the constraint \eqref{equ:r} is equivalent to $H^{-1}M_0=H^{-1}L_0+S_0$ where $H^{-1}L_0$ is also low rank, so we can still resort to the classical matrix separation problem:
\begin{equation}\label{equ:pinv}
\left\{
\begin{array}{ll}
(\hat S_1, \hat Y)=\argmin_{S, Y}\gamma\|S\|_1+\|Y\|_*, \quad\text{subject to }S+Y=H^{-1}M_0\\
\hat L_1=H\hat Y.
\end{array}\right.
\end{equation}
It is worth noting that $(\hat S, \hat L)$ from \eqref{equ:p} is not necessarily the same as $(\hat S_1, \hat L_1) $ from \eqref{equ:pinv}, essentially due to $\|HY\|_*\neq\|Y\|_*$, so investigating the conditions of recovering $(S_0, L_0)$ via \eqref{equ:p}  is of independent interest even when $H$ is invertible.

\subsection{Applications}
The setup in \eqref{equ:r} is partially motivated by the electrodermal activity (EDA) signal decomposition problem. EDA assesses the naturally occurring changes in electrical properties of human skin. The EDA signal includes both
background tonic (skin conductance level: SCL) and rapid phasic components (Skin Conductance Responses: SCRs) that result from sympathetic neuronal activity. EDA is arguably the most useful index of changes in sympathetic arousal that are tractable to emotional and cognitive states as it is the only autonomic psychophysiological variable that is
not contaminated by parasympathetic activity~\cite{edacs, edaguide}. Specifically, the observed EDA signal $y$ has the form
\begin{equation}\label{equ:eda}
y=t+r+e,
\end{equation}
where $t$ is SCL, the tonic component, $r$ is SCR, the phasic component, and $e$ is a noise component. The signal $r$ can be further modeled as a convolution of some filter $h$ and the SCR events $x$, that is,
\begin{equation}
r=h*x=Hx.
\end{equation}
In most prior literature, the SCR event $x$ is modeled to be sparse and have nonnegative coordinates~\cite{chaspari2014sparse}.

The goal of EDA decomposition is to recover $x$ from $y$.
Let $y_1, y_2, \cdots, y_n\in\R^m$ be $m$ observed EDA signals. Let $y_i=t_i+Hx_i+e_i$ following \eqref{equ:eda}, then we have
$$Y=T+HX+E,$$
where $Y=[y_1,\cdots,y_n], T=[t_1, \cdots, t_n], X=[x_1,\cdots,x_n], E=[e_1,\cdots,e_n]$.

$T$ will be close to a low rank matrix (rank 1) and $X$ is a sparse matrix. More details on $H$ in this setting can be found in Section \ref{sec:num_eda}.


We expect our general setting go beyond this motivating example. For instance, we are given a blurred surveillance video $M = HS_0+HL_0$ where $H$ is the blurring matrix, $S_0$ represents the moving objects (sparse), and $L_0$ stores the constant background. The  recovery of $S_0$ means removing background  and deblurring moving objects simultaneously.

\subsection{Related Work and Contributions}
The seminal work \cite{C11, CLMW11} largely laid theoretical foundation for separability via \eqref{equ:p} when $H$ is the identity matrix.  This is often referred as the robust principal component analysis (PCA) problem. In \cite{C11}, Chandrasekaran et. al. provided a deterministic result on sufficient conditions of $L_0$ and $S_0$ for separability. The results in \cite{CLMW11} are not deterministic. In particular, it assume a random sparse pattern for $S_0$. It however has a better theoretical guarantee than that in \cite{C11}. The work \cite{HKZ11} improves the results in \cite{C11} by allowing weaker sufficient conditions, and it also studies the effect of perturbations on the accuracy of the recovered components.

As far as the authors can tell, a matrix separation problem in the general form of \eqref{equ:r} has not been studied before. The work \cite{YGL15} is related in that the mask matrix is being multiplied to the low rank matrix. We present a first theoretical guarantee (Theorem \ref{thm:main}) for the successful recovery of the low rank and sparse components through the convex programming \eqref{equ:r}. The result is deterministic. Our analysis follows the work \cite{C11} closely. However, the involvement of $H$ poses nontrivial challenges as we will explain. We propose a reasonable and minimal condition on $H$, which requires it  to satisfy a restricted infinity norm property (Definition \ref{def:rinp}). 
Several numerical experiments are conducted in Section \ref{sec:num} with various $H$.


\section{Main Results}
\subsection{Notations and Preliminaries}

For a vector $x$, we let $\|x\|_p=(\sum_{i}|x_i|^p)^{1/p}$ for $p\geq1$ and $\|x\|_\infty=\max_i|x_i|$. We also define $\sign(x)$ coordinate-wise as
$$\sign(x)_i = \left\{\begin{array}{ll}
1, &\text{ if }x_i >0\\
-1, &\text{ if }x_i<0\\
0, &\text{ if }x_i=0\end{array}\right..$$
$\sign()$ of a matrix is defined similarly.

For a matrix $A=(a_{ij})$, $\|A\|_\infty=\max_{i,j}|a_{ij}|$, and $\|A\|_1=\sum_{i,j}|a_{ij}|$ are its vectorized infinity norm and 1-norm respectively.

 The nuclear norm $\|A\|_*$ is the sum of its singular values, also called the Schattern 1-norm. Moreover, we let $\|A\|:=\sup_{\|x\|_2=1}\|Ax\|_2$ be its spectral norm, and
$$\|A\|_\mi:=\sup_{\|x\|_\infty=1}\|Ax\|_\infty$$
be its ``spectral infinity norm".
It is known that $\|A\|_\mi=\max_{i}\sum_j|a_{ij}|$, the maximum absolute row sum. Given this notation, we then have
\begin{equation}\label{equ:AM}
\|AM\|_\infty\leq\|A\|_\mi\|M\|_\infty, \text{ for all }M.
\end{equation} 

For two real matrices $A, B$ of the same dimension, define their inner product to be $\langle A, B\rangle=\tr(A^TB)$. Using dual norm properties, we have
$$|\langle A, B\rangle|\leq\|A\|_*\|B\|, \quad |\langle A, B\rangle|\leq\|A\|_1\|B\|_\infty.$$

Moreover, we use $A[i, :]$  and $A[:, j]$ to indicate the $i$th row and $j$th column of $A$ respectively.

For any subspace $V$, $V^\perp$ is its orthogonal complement and $P_V$ is the orthogonal projection onto $V$.
We use $\ran(A)$ to denote the range of $A$, i.e., the column space of $A$. 

The set
$$\Tc(A):=\{X+Y: \ran(X)\subset\ran(A), \ran(Y^T)\subset\ran(A^T)\}$$
is the tangent space at $A$ with respect to the variety of all matrices with rank less than or equal to $\rank(A)$~\cite{C11}. It can be computed that the projection onto this tangent space is
\begin{equation}\label{equ:pt}
P_{\Tc(A)}(X) = UU^TX + XVV^T-UU^TXVV^T,
\end{equation}
where $A=U\Sigma V^T$ is its reduced singular value decomposition (SVD), that is, both $U$ and $V$ have $\rank(A)$ columns.
Given  \eqref{equ:pt}, we can compute
$$P_{\Tc(A)^\perp}(X) = (I-UU^T)X(I-VV^T),$$
and therefore 
$$\|P_{\Tc(A)^\perp}(X)\|\leq\|X\|.$$
Consequently,  $\|P_{\Tc(A)}(X)\|\leq \|X-P_{\Tc(A)^\perp}(X)\|\leq 2\|X\|.$ 

For any matrix $A$, define
$$\Omega(A) :=  \text{the set of matrices whose support is contained within the support of }A.$$

For a matrix $A\in\R^{m\times n}$, let $d_r(A), d_c(A)$ be the maximum number of nonzero entries per row, per column respectively.  Specifically,
$$d_r(A)=\max_{i\in\{1, 2, \cdots, m\}}\{|\{j: A_{ij}\neq0\}|\}$$
where $|\cdot|$ is the cardinality of a set, and $d_c(A)$ is defined likewise.
Then $$d(A):=\max\{d_r(A), d_c(A)\}$$ is the maximum number of nonzero entries per row or column of $A$.

With $H\in\R^{m\times p}$, we let $H^\dagger\in \R^{p\times m}$ be its pseudoinverse (Moore-Penrose inverse). This means $HH^\dagger$ is the orthogonal projection onto the column space of $H$ and $H^\dagger H$ is the orthogonal projection onto the row space of $H$. We have
$$HH^\dagger H = H, \text{ or } H^\dagger HH^\dagger = H^\dagger.$$
Let $H=U_H\Sigma_HV_H^T$ be its reduced SVD so that $U_H\in\R^{m\times K}, \Sigma_H\in\R^{K\times K}, V_H\in\R^{p\times K}$, where $K=\rank(H)$. The pseudoinverse can be expressed as $H^\dagger=V_H\Sigma_H^{-1}U_H^T$.

For an integer $n$, we let $[n]$ to denote the index set $\{1, 2, \cdots, n\}$.

\subsection{Separable Conditions}
In order to establish conditions for successful recovery via \eqref{equ:p}, we start by reviewing the case when $H$ is the identity matrix. When recovering $S_0$ and $L_0$ from $M_0=S_0 + L_0$, the sparse matrix $S_0$ itself should not be low rank in order to be distinguished from $L_0$. In \cite{C11}, Chandrasekaran et al. defines the quantity
\begin{equation}
\mu(S):=\max_{A\in\Omega(S), \|A\|_\infty\leq1}\|A\|
\end{equation}
for this task. The idea is that with an infinity norm restriction, a matrix tends to reach its maximal spectral norm while being low rank. For example, with restriction $\|A\|_\infty\leq1$ only, $\|A\|$ is biggest when $A$ is the all-ones matrix.  With the additional restriction $A\in\Omega(S)$, placing an appropriate upper bound on $\mu(S)$ forces matrices in $\Omega(S)$ to have reasonable rank.
For our masked matrix separation problem, we ask, more appropriately, that $HS_0$ not to be low rank, which motivates the following quantity
\begin{equation}\label{equ:muH}
\mu_H(S):=\max_{A\in\Omega(S), \|A\|_\infty\leq1}\|HA\|.
\end{equation}

For the regular matrix separation problem, the work \cite{C11} also defines
\begin{equation}\label{equ:xi}
\xi(L):=\max_{A\in \Tc(L), \|A\|\leq1}\|A\|_\infty.
\end{equation}
Small $\xi(L)$ implies $L$ is not too sparse, and thus avoid having $L$ to be low rank and sparse simultaneously.
To adapt to our masked setting, consider the quantity $\max_{HA\in \Tc(L), \|HA\|\leq1}\|A\|_\infty$, which at first glance appears to be a reasonable modification to \eqref{equ:xi} as we do not want a sparse $A$ such that $HA$ is in $\Tc(L)$. 
However, this is problematic when $H$ has a non-trivial null space, in which case $\sup_{HA\in \Tc(L), \|HA\|\leq1}\|A\|_\infty=\infty$ by choosing nontrivial $A$ such that $HA=0$.
For this reason, we further modify $\max_{HA\in \Tc(L), \|HA\|\leq1}\|A\|_\infty$  to be 
$\max_{HA\in \Tc(L), \|HA\|\leq1}\|H^T HA\|_\infty,$ which is why we define
\begin{equation}\label{equ:xiH}
\xi_H(L):=\max_{B\in \Tc(L), \|B\|\leq1}\|H^T B\|_\infty.
\end{equation}


\begin{definition}\label{def:rinp}
Given a matrix $S_0\in \R^{p\times n}$ and $0<\delta <1$, we say a matrix $G$ of dimension $m\times p$ has the $S$-$\delta$-\emph{restricted infinity norm property} ($S_0$-$\delta$-RINP) if 
\begin{equation}\label{equ:Hrip}
\|(I-G^TG)A\|_\infty\leq\delta\|A\|_\infty \text{ for all }A\in\Omega(S_0).
\end{equation}
\end{definition}

Ultimately, we would like our mask $H$ to satisfy \eqref{equ:Hrip} up to a \textbf{column scaling}.
\begin{definition}\label{def:rinps}
Given a matrix $S_0\in \R^{p\times n}$ and $0<\delta <1$, we say a matrix $H$ of dimension $m\times p$ has the \emph{scaled-$S$-$\delta$-restricted infinity norm property} (scaled-$S_0$-$\delta$-RINP) if there exists $G=HD$, where $D$ is an invertible diagonal matrix such that \eqref{equ:Hrip} holds.
\end{definition}
To reiterate, $H$ satisfies  the scaled-$S_0$-$\delta$-RINP if 
\begin{equation}\label{equ:Grip}
\begin{array}{l}
\exists\ G=HD\text{ where }D\text{ is diagonal and invertible such that}\\
 \|(I-G^TG)A\|_\infty\leq\delta\|A\|_\infty \text{ for all }A\in\Omega(S_0).
\end{array}
\end{equation}

Given $S_0$ is sparse,  \eqref{equ:Hrip} is quite reminiscent of the restricted isometry property~\cite{C08} used in the compressed sensing literature, but with an infinity norm.

We see that \eqref{equ:Grip} with $\delta<1$ implies $G$ is invertible on the set $\Omega(S_0)$ (this is equivalent to $H$ being invertible on $\Omega(S_0)$). This is because otherwise we can find an $A$ such that $\eqref{equ:Hrip}$ is violated. 
More analysis can be found in Section \ref{sec:rinp}.

It is easy to see that we need $H\Omega(S_0)\cap \Tc(L_0)=\{0\}$ for unique identifiability of $S_0, L_0$ from \eqref{equ:r}. The following proposition provides a sufficient condition for this. 

\begin{proposition}\label{pro:first}
Given any $L_0\in\R^{m\times n}, S_0\in\R^{p\times n}$, and let  $H$ satisfy the scaled $S_0$-$\delta$-RINP \eqref{equ:Grip}.
If $\mu_G(S_0)\xi_G(L_0)<1-\delta$, then $H\Omega(S_0)\cap \Tc(L_0)=\{0\}$.
\end{proposition}
\begin{proof}
Since $H\Omega(S_0)=G\Omega(S_0)$, we need to prove $G\Omega(S_0)\cap \Tc(L)=\{0\}$.
By way of contradiction suppose we have nonzero $GA_0\in G\Omega(S_0)\cap \Tc(L_0)$ where $A_0\in\Omega(S_0)$, then 
\begin{equation}\label{equ:p2}
\mu_G(S_0)\geq\Big\|\frac{GA_0}{\|A_0\|_\infty}\Big\|=\frac{\|GA_0\|}{\|A_0\|_\infty},
\end{equation} and similarly
\begin{equation}\label{equ:p3}
\xi_G(L_0)\geq\Big\|\frac{G^T GA_0}{\|GA_0\|}\Big\|_\infty=\frac{\|G^TGA_0\|_\infty}{\|GA_0\|}{\geq}\frac{(1-\delta)\|A_0\|_\infty}{\|GA_0\|}.
\end{equation}
We get a contradiction when multiplying inequalities \eqref{equ:p2} and \eqref{equ:p3}, so we must have $\mu_G(S_0)\xi_G(L_0)<1-\delta$.
\end{proof}


Following standard notation in convex analysis~\cite{R70}, for a continuous function $f:\R^N\rightarrow\R$, its \emph{subdifferential} at $x\in \R^N$ is defined as
$$\partial f(x)=\{x^*: f(y)\geq f(x)+\langle x^*, y-x\rangle \text{ for any }y\in\R^N\}.$$
We can characterize the subdifferential of the $\ell_1$ norm. We have $y\in\partial \|x\|_1$ if and only if 
\begin{equation}\label{equ:subl1}
y_{\supp(x)}=\sign(x) \text{ and }\|y_{(\supp(x))^c}\|_\infty\leq1.
\end{equation}
The subdifferential of the nuclear norm has a form similar to \eqref{equ:subl1}~\cite{W92}, and we have $A\in\partial\| B\|_*$ if and only if
\begin{equation}\label{equ:subnuc}
P_{\Tc(B)}(A)=U_BV_B^T, \text{ where }B=U_B\Sigma V_B^T\text{ is its reduced SVD, and} \|P_{\Tc(B)^\perp}(A)\|\leq1.
\end{equation}

Using basic convex optimization, we know $(S_0, L_0)$ is a minimizer of \eqref{equ:p} if and only if
\begin{equation}\label{equ:dual}
\text{ there exists a matrix }Q\text{ such that }H^TQ\in\gamma\partial\| S_0\|_1\text{ and }Q\in\partial\| L_0\|_*.
\end{equation} 
Using \eqref{equ:subl1} and \eqref{equ:subnuc},  the optimality condition \eqref{equ:dual} becomes
\begin{equation}\label{equ:opt}
\text{There exists }Q \text{ such that }\left\{
\begin{array}{ll}
P_{\Omega(S_0)}(H^TQ)=\gamma\sign(S_0), \|P_{\Omega(S_0)^\perp}(H^TQ)\|_\infty\leq\gamma,\\
P_{\Tc(L_0)}(Q)=UV^T, \|P_{\Tc(L_0)^\perp}(Q)\|\leq1,
\end{array}\right.
\end{equation}
given the SVD of $L_0$ is $L_0 = U\Sigma V^T$.

The following proposition (where we use the notation $G$ instead of $H$) can be viewed as a generalization of  \cite[Proposition 2]{C11} or  \cite[Lemma 2.4]{CLMW11}.


\begin{proposition}\label{pro:dual}
Let $L_0=U\Sigma V^T$ be its SVD and $M_0=GS_0+L_0$. If the following conditions (1) and (2) hold, then for any optimizer $(\hat S, \hat L)$ of 
$$(\hat S, \hat L)=\argmin_{S, L}\gamma\|S\|_1+\|L\|_*, \quad\text{subject to }L+GS=M_0,
$$ it must hold that $G\hat S=GS_0$ and $\hat L=L_0$.
\begin{itemize}
\item[(1)] $G\Omega(S_0)\cap \Tc(L_0)=\{0\}$.
\item[(2)] There exists $Q$ such that
\begin{enumerate}
\item[(a)] $P_{\Tc(L_0)}(Q)=UV^T$, where $U\Sigma V^T$ is the singular value decomposition of $L_0$.
\item[(b)] $\|P_{(\Tc(L_0))^\perp}(Q)\|<1$
\item[(c)] $P_{\Omega(S_0)}(G^TQ)=\gamma\sign(S_0)$
\item[(d)] $\|P_{(\Omega(S_0))^\perp}(G^TQ)\|_\infty<\gamma$.
\end{enumerate}
\end{itemize}
\end{proposition}

\begin{proof}
By \eqref{equ:opt}, condition (2) guarantees $(S_0, L_0)$ is a minimizer of \eqref{equ:p}. Suppose $(S_0+N_S,L_0+N_L)$ is also a minimizer of \eqref{equ:p}, then 
\begin{equation}\label{equ:0}
GN_S+N_L=0.
\end{equation}

Let $(Q_S, Q_L)$ be any subgradient of $f(S,L)=\gamma\|S\|_1+\|L\|_*$ at $(S_0, L_0)$, then we have
\begin{align}
&f(S_0+N_S,L_0+N_L)\geq f(S_0, L_0)+\langle (Q_S, Q_L),(N_S, N_L)\rangle\\
\Longleftrightarrow\ &\gamma\|S_0+N_S\|_1+\|L_0+N_L\|_*\geq\gamma\|S_0\|_1+\|L_0\|_*+\langle Q_S, N_S\rangle+ \langle Q_L, N_L\rangle.
\end{align}

Let $\Omega=\Omega(S_0), \Tc = \Tc(L_0)$ for ease of notation.

Since $Q_S\in\gamma\partial\|S_0\|_1$ and $Q_L\in\partial\|L_0\|_1$, we have
$$Q_S=\gamma\sign(S_0)+P_{\Omega^\perp}(Q_S)\text{ with }\|P_{\Omega^\perp}(Q_S)\|_\infty\leq\gamma$$
$$Q_L=UV^T+P_{\Tc^\perp}(Q_L)\text{ with }\|P_{\Tc^\perp}(Q_L)\|\leq1$$

Similarly by condition (2), we have
$$Q=UV^T+P_{\Tc^\perp}(Q), G^TQ=\gamma\sign(S_0)+P_{\Omega^\perp}(G^TQ).$$ Then
\begin{align}\notag
\langle Q_L, N_L\rangle&=\langle UV^T+P_{\Tc^\perp}(Q_L), N_L\rangle=\langle Q-P_{\Tc^\perp}(Q)+P_{\Tc^\perp}(Q_L), N_L\rangle\\\label{equ:QL}
&= \langle P_{\Tc^\perp}(Q_L)-P_{\Tc^\perp}(Q), N_L\rangle+\langle Q, N_L\rangle
\end{align}
and
\begin{align}\notag
\langle Q_S, N_S\rangle&=\langle \gamma \sign(S_0)+P_{\Omega^\perp}(Q_S), N_S\rangle\\\notag
&=\langle G^TQ-P_{\Omega^\perp}(G^TQ)+P_{\Omega^\perp}(Q_S), N_S\rangle\\\label{equ:QS}
&=\langle P_{\Omega^\perp}(Q_S)-P_{\Omega^\perp}(G^TQ), N_S\rangle+\langle Q, GN_S\rangle
\end{align}

We will pick $Q_S$ such that $P_{\Omega^\perp}(Q_S)=\gamma\sign(P_{\Omega^\perp}(N_S))$, and pick
 $Q_L$ such that $P_{\Tc^\perp}(Q_L)=U_1V_1^T$ where $P_{\Tc^\perp}(N_L)=U_1\Sigma_1 V_1^T$ is its SVD.

So by \eqref{equ:0}, \eqref{equ:QL}, \eqref{equ:QS}, and the pick of $Q_S, Q_L$,
\begin{align*}
&\langle Q_L, N_L\rangle+\langle Q_S, N_S\rangle\\
= & \langle P_{\Tc^\perp}(Q_L)-P_{\Tc^\perp}(Q), N_L\rangle+\langle P_{\Omega^\perp}(Q_S)-P_{\Omega^\perp}(G^TQ), N_S\rangle+\langle Q, N_L+GN_S\rangle\\
= & \langle P_{\Tc^\perp}(Q_L)-P_{\Tc^\perp}(Q), P_{T^\perp}(N_L)\rangle+\langle P_{\Omega^\perp}(Q_S)-P_{\Omega^\perp}(G^TQ), P_{\Omega^\perp}(N_S)\rangle+0\\
= &\|P_{\Tc^\perp}(N_L)\|_*-\langle P_{\Tc^\perp}(Q), P_{\Tc^\perp}(N_L)\rangle+\gamma\|P_{\Omega^\perp}(N_S)\|_1-\langle P_{\Omega^\perp}(G^TQ), P_{\Omega^\perp}(N_S)\rangle\\
\geq & \|P_{\Tc^\perp}(N_L)\|_*(1- \|P_{\Tc^\perp}(Q)\|) + \|P_{\Omega^\perp}(N_S)\|_1(\gamma-\|P_{\Omega^\perp}(G^TQ)\|_\infty)
\end{align*}

Given condition (2)-(b)(d), the only way $f(S_0+N_S,L_0+N_L)= f(S_0, L_0)$ is that $\|P_{\Tc^\perp}(N_L)\|_*=0$ and $\|P_{\Omega^\perp}(N_S)\|_1=0$, which means
$N_L\in \Tc$ and $N_S\in\Omega$, so $GN_S=-N_L\in G\Omega\cap \Tc$ implying $GN_S=N_L=0$ by condition (1).

\end{proof}

\begin{corollary}
Let $L_0=U\Sigma V^T$ be its SVD and $M_0=GS_0+L_0$ where $G$ satisfies $S_0$-$\delta$-RINP \eqref{equ:Hrip}. If the following conditions (1)(2) hold, then $(S_0, L_0)$ must be the unique minimizer of $$(\hat S, \hat L)=\argmin_{S, L}\gamma\|S\|_1+\|L\|_*, \quad\text{subject to }L+GS=M_0.
$$
\begin{itemize}
\item[(1)] $G\Omega(S_0)\cap \Tc(L_0)=\{0\}$.
\item[(2)] There exists $Q$ such that
\begin{enumerate}
\item[(a)] $P_{\Tc(L_0)}(Q)=UV^T$, where $U\Sigma V^T$ is the singular value decomposition of $L_0$.
\item[(b)] $\|P_{(\Tc(L_0))^\perp}(Q)\|<1$
\item[(c)] $P_{\Omega(S_0)}(G^TQ)=\gamma\sign(S_0)$
\item[(d)] $\|P_{(\Omega(S_0))^\perp}(G^TQ)\|_\infty<\gamma$
\end{enumerate}
\end{itemize}
\end{corollary}
\begin{proof}

Following the proof of Proposition \ref{pro:dual}, we have $GN_S = 0$. By the comments following \eqref{equ:Grip}, we know $G$ is invertible on $\Omega(S_0)$, and hence $N_S=0$.
\end{proof}

We need the following lemma to prove the main theorem. It is only used once but having it stand alone provides a better flow and  clarify of the main proof. 

\begin{lemma}\label{lem:AB}
Let $X, Y$  be two finite dimensional Hilbert spaces and $G$ be a linear operator from $X$ to $Y$. Let $A\subset X, B\subset Y$ be two subspaces such that $GA\subset Y$, $GA\cap B=\{0\}$, and $G$ is invertible on $A$.
For any $a\in A, b\in B$, there exists $w\in Y$ such that $P_A(G^Tw)=a, P_B(w)=b$.
\end{lemma}
\begin{proof}
Note that if we are able to find $z\in Y$ such that 
\begin{equation}\label{equ:z}
P_A(G^TP_{B^\perp}z)=a-P_A(G^Tb),
\end{equation} 
then $w=P_{B^\perp}z+b\in Y$ would be the desired $w$. This is because it follows $P_A(G^Tw)=P_A(G^TP_{B^\perp}z+G^Tb)=a-P_A(G^Tb)+P_A(G^Tb)=a$.

Suppose there is not such a $z\in Y$ satisfying \eqref{equ:z}, then we can conclude that $P_A(G^TP_{B^\perp}Y)$ must be a proper subspace of $A$, so there exists $c\neq0, c\in A$ such that $c\perp G^TP_{B^\perp}Y$. For any $y\in Y$, we have
$$\langle Gc, P_{B^\perp}y\rangle=\langle c, G^TP_{B^\perp}y\rangle=0,$$
which means $Gc\in B$.

$Gc\in GA\cup B$, implying $Gc = 0$ and further $c=0$ given $G$ is invertible on $A$, a contradiction.

\end{proof}

%
%

\begin{theorem}\label{thm:main}
Given $M_0=HS_0+L_0$ where  $H$ satisfies  \eqref{equ:Grip} with $\delta<1/3$, and $G$ be the column-scaled $H$ in \eqref{equ:Grip}.
If
\begin{equation}\label{equ:incoherence}
\mu_G(S_0)\xi_G(L_0)<\frac{1-3\delta}{6},
\end{equation}
then there exists $\gamma>0$ such that for any optimizer $(\hat S, \hat L)$ of \eqref{equ:p}, we must have $\hat S=S_0, \hat L=L_0$.  
\end{theorem}
\begin{proof}
Let $G=HD$ as usual.
We see that $\frac{1-3\delta}{6}<1-\delta$, so the assumption of Proposition \ref{pro:first} is satisfied.
With Proposition \ref{pro:dual}, it is sufficient to find a matrix $Q$ that satisfies (2)(a)-(d) in Proposition \ref{pro:dual}. We again let  $\Omega=\Omega(S_0), \Tc = \Tc(L_0)$, and $L_0=U\Sigma V^T$ be its SVD.

We will utilize Lemma \ref{lem:AB} where we let $A=\Omega, B=\Tc, X = \Omega+G^{-1}B, Y=G\Omega+\Tc$ (recall  $G^{-1}B$ is the pre-image of $B$ under $G$.). 
We have $GA\subset Y$ and $B\subset Y$. We also have $GA\cap B=\{0\}$ by Proposition \ref{pro:first}.
The conclusion of Lemma \ref{lem:AB} guarantees the existence of $Q\in G\Omega+\Tc$ such that $P_{\Omega}(G^TQ)=\gamma\sign(S_0)$ and $P_\Tc(Q)=UV^T$. 
The rest of the proof is to make sure there exists $\gamma$ such that 
\begin{equation}\label{equ:bd}
\|P_{\Tc^\perp}Q\|<1, \|P_{\Omega^\perp}(G^TQ)\|_\infty<\gamma.
\end{equation}

Since $Q\in G\Omega+\Tc$, we let $Q=GQ_\Omega+Q_T$ for some $Q_\Omega\in\Omega, Q_T\in\Tc$.


Let $Q_\Omega=\gamma\sign(S_0)+\epsilon_\Omega$ and $Q_T=UV^T+\epsilon_T$. So
$$ Q=GQ_\Omega+Q_T=GQ_\Omega+UV^T+\epsilon_T.$$
Applying the projection $P_\Tc$ to the equation $Q=GQ_\Omega+UV^T+\epsilon_T$, we get 
\begin{equation}\label{equ:eTg}
UV^T=P_\Tc(GQ_\Omega)+UV^T+\epsilon_T\quad\Longrightarrow \epsilon_T=-P_\Tc(GQ_\Omega).
\end{equation}
Similarly, we apply the projection $P_\Omega$ to $G^TQ=G^TG(\gamma\sign(S_0)+\epsilon_\Omega)+G^TQ_T$ and obtain
\begin{equation}\label{equ:eOg}
\gamma\sign(S_0)=P_\Omega G^TG(\gamma\sign(S_0)+\epsilon_\Omega)+P_\Omega(G^TQ_T).
\end{equation}
Using definitions of $\mu_G$ and $\xi_G$,
\begin{align}\label{equ:H-TQg}
&\|GQ_\Omega\|\leq\mu_G(S_0)\|Q_\Omega\|_\infty\leq\mu_G(S_0)(\gamma+\|\epsilon_\Omega\|_\infty),\\
\label{equ:HTQ}
&\| G^TQ_T\|_\infty\leq\xi_G(L_0)\|Q_T\|
\leq  \xi_G(L_0)(1+\|\epsilon_T\|).
\end{align}
Denote $e= \xi_G(L_0)$ and $u=\mu_G(S_0)$, then using \eqref{equ:eTg}, 
\begin{equation}\label{equ:eTe}
\|\epsilon_T\|=\|P_\Tc(GQ_\Omega)\|\leq 2\|GQ_\Omega\|\stackrel{\eqref{equ:H-TQg}}{\leq}2u(\gamma+\|\epsilon_\Omega\|_\infty).
\end{equation}
The RINP condition on $G$   implies for $A\in\Omega$,
\begin{equation}\label{equ:Hrip2}
\|(P_\Omega G^TG-I)A\|_\infty=\|P_\Omega (G^TG-I)A\|_\infty\leq\|(G^TG-I)A\|_\infty\leq\delta\|A\|_\infty,
\end{equation}
and
\begin{align}
\|P_{\Omega^\perp}G^TGA\|_\infty&=\|G^TGA-P_\Omega G^TGA\|_\infty=\|G^TGA-A-(P_\Omega G^TGA-A)\|_\infty\\
\label{equ:op}&\leq \|(G^TG-I)A\|_\infty+ \|(P_\Omega G^TG-I)A\|_\infty\leq 2\delta\|A\|_\infty.
\end{align}
So given  $\epsilon_\Omega\in\Omega$,
\begin{align}\notag
\|\epsilon_\Omega\|_\infty&\stackrel{\eqref{equ:Hrip2}}{\leq}\frac{1}{1-\delta}\|P_\Omega G^TG\epsilon_\Omega\|_\infty\\\notag
&\stackrel{\eqref{equ:eOg}}{\leq}\frac{1}{1-\delta}\left(\|(I-P_\Omega G^TG)\gamma\sign(S_0)\|_\infty+\|P_\Omega G^TQ_T\|_\infty\right)\\
&\leq \frac{1}{1-\delta}(\delta\|\gamma\sign(S_0)\|_\infty+\| G^TQ_T\|_\infty)\\
\label{equ:eOe}
&\stackrel{\eqref{equ:HTQ}}{\leq}\frac{1}{1-\delta}(\delta\gamma+e(1+\|\epsilon_T\|))
\end{align}

Combining \eqref{equ:eTe} and \eqref{equ:eOe} yields
$$\|\epsilon_\Omega\|_\infty\leq\frac{\delta\gamma}{1-\delta}+\frac{e}{1-\delta}(1+2u\gamma+2u\|\epsilon_\Omega\|_\infty)\Rightarrow (1-\frac{2ue}{1-\delta})\|\epsilon_\Omega\|_\infty\leq\frac{\delta\gamma+e(1+2u\gamma)}{1-\delta},$$
so
\begin{equation}\label{equ:etog}
\|\epsilon_\Omega\|_\infty\leq\frac{\delta\gamma+e(1+2u\gamma)}{1-\delta-2ue}, \quad\|\epsilon_T\|\leq\frac{2u(\gamma+e)}{1-\delta-2ue}.
\end{equation}
Now
$$
\|P_{\Tc^\perp} Q\|=\|P_{\Tc^\perp}(GQ_\Omega)\|\leq\|GQ_\Omega\|\stackrel{\eqref{equ:H-TQg}}{\leq} u(\gamma+\|\epsilon_\Omega\|_\infty)\stackrel{\eqref{equ:etog}}{\leq}\frac{u(\gamma+e)}{1-\delta-2ue}.
$$ 
We have
\begin{equation}\label{equ:gammau}
\frac{u(\gamma+e)}{1-\delta-2ue}<1\Longleftrightarrow u\gamma+ue<1-\delta-2ue\Longleftrightarrow\gamma<\frac{1-\delta-3ue}{u},
\end{equation}
so having $\gamma<\frac{1-\delta-3ue}{u}$ ensures the first inequality of \eqref{equ:bd}. 
Furthermore,
\begin{align*}
\|P_{\Omega^\perp}(G^TQ)\|_\infty&=\|P_{\Omega^\perp}G^TG(\gamma\sign(S_0)+\epsilon_\Omega)+P_{\Omega^\perp}( G^TQ_T)\|_\infty\\
&\stackrel{\eqref{equ:op}}{\leq}2\delta(\gamma+\|\epsilon_\Omega\|_\infty)+\|P_{\Omega^\perp}( G^TQ_T)\|_\infty
\stackrel{\eqref{equ:HTQ}}{\leq}2\delta(\gamma+\|\epsilon_\Omega\|_\infty)+e(1+\|\epsilon_T\|)\\
&\stackrel{\eqref{equ:etog}}{\leq}\frac{2\delta(\gamma+e)}{1-\delta-2ue}+\frac{e(1-\delta-2ue+2u\gamma+2ue)}{1-\delta-2ue}=\frac{e+\delta e+2ue\gamma+2\delta\gamma}{1-\delta-2ue}.
\end{align*}
To ensure the second inequality of \eqref{equ:bd}, we require
\begin{align}\notag
\frac{e+\delta e+2ue\gamma+2\delta\gamma}{1-\delta-2ue}<\gamma&\Longleftrightarrow e+\delta e+2ue\gamma+2\delta\gamma<\gamma -\gamma\delta-2ue\gamma\\\notag
&\Longleftrightarrow (1+\delta)e<\gamma(1-4ue-3\delta)\\\label{equ:gammal}
&\Longleftrightarrow\gamma>\frac{(1+\delta)e}{1-3\delta-4ue}.
\end{align}

Given \eqref{equ:gammau} and \eqref{equ:gammal}, we require
\begin{align*}
&\frac{(1+\delta)e}{1-3\delta-4ue}<\frac{1-\delta-3ue}{u}\Longleftrightarrow (1+\delta)ue<(1-\delta)(1-3\delta)+12u^2e^2-3ue(1-3\delta)-4ue(1-\delta)\\
&\Longleftrightarrow 12u^2e^2-(8-12\delta)ue+(1-\delta)(1-3\delta)>0 \Longleftrightarrow (2ue-(1-\delta))(6ue-(1-3\delta))>0,
\end{align*}
which can be satisfied if $ue<\frac{1-3\delta}{6}$.
\end{proof}

\begin{remark}
The analysis in Theorem \ref{thm:main} uses similar methods as in \cite{C11}. However, the general setting with $H$ poses extra complication and challenge.
Specifically, for $A\in\Omega(S_0)$, we do not necessarily have $G^TGA\in \Omega(S_0)$. This is essentially why RINP is needed.

\end{remark}

\subsection{Restricted Infinity Norm Property}\label{sec:rinp}
 \begin{lemma}\label{lem:delta}
 Let $d_c(S_0)=d$, that is, each column of $S_0$ has at most $d$ nonzero entries, then $G$ has the  $S_0$-restricted infinity norm property \eqref{equ:Hrip} with $\delta=\max_{|T|\leq d}\|P_T\|_\mi$ where $P=I-G^TG$. \end{lemma}
 \begin{proof}
 Let $a_j$ be the $j$th column of $A$.
 
 $\max_{\|A\|_\infty\leq1, A\in\Omega(S_0)}\|PA\|_\infty=\max_{\|A\|_\infty\leq1, A\in\Omega(S_0)}\max_{j}\|Pa_j\|_\infty\leq\max_j\max_{\|a_j\|_\infty\leq1, |\text{supp}(a_j)|\leq d}\|Pa_j\|_\infty=\max_j\max_{\|a\|_\infty\leq1, |T|\leq d}\|P_Ta\|_\infty=\max_{\|a\|_\infty\leq1, |T|\leq d}\|P_Ta\|_\infty=\max_{|T|\leq d}\|P_T\|_\mi$.
 \end{proof}
 
Now we   analyze the feasibility of RINP \eqref{equ:Hrip} or \eqref{equ:Grip}  for three special cases of $H$.

\subsubsection{When columns of $H$ are orthogonal}
In this case, we can find $D$ such that $(HD)^T(HD) = I$ so $H$ satisfies \eqref{equ:Hrip} with $\delta =0$. Note that this is when the rows of $HD$ form a Parseval frame of $\R^p$. $H$ is shaped tall and thin with full column rank.
\begin{corollary}\label{cor:delta0}
Given $M_0=HS_0+L_0$ where there exists an invertible diagonal matrix $D$ such that $(HD)^THD=I$. If $\mu_{HD}(S_0)\xi_{HD}(L_0)<\frac{1}{6},$
then there exists $\gamma>0$ such that for any optimizer $(\hat S, \hat L)$ of \eqref{equ:p}, we must have $\hat S=S_0, \hat L=L_0$.  
\end{corollary}

Note that Corollary \ref{cor:delta0} reduces to \cite[Theorem 2]{C11} exactly when $H$ is the identity matrix. 
%
%

\subsubsection{When $\|H\|\|H^\dagger\|=1$}

This translates to the biggest singular values and the smallest nonzero singular values  of $H$ are equal, or $H$ has condition number 1 ($H$ can be singular and of any shape). The implication is
$H=\|H\|UV^T$ is its SVD with $U\in\R^{m\times k}$ and $V\in\R^{p\times k}$ where $k=\rank(H)$ .

In this case, we will let $D=\frac{1}{\|H\|}I$ which leads to 
\begin{equation}
I-(HD)^T (HD) = I-VV^T=P_{\ker (H)}
\end{equation}


Let $\otimes$ be the kronecker product and  $\vecc(\cdot)$ be the vectorization of a matrix. If one requires a projection $P$ to satisfy $\|PA\|_\infty\leq\delta\|A\|_\infty$ for all $A$, which is $\|(I\otimes P)\vecc(A)\|_\infty\leq\delta\|\vecc(A)\|_\infty$, then $\delta=\|I\otimes P\|_\mi=\|P\|_\mi$, which is at least 1 for any orthogonal projection $P$ given $P=P^2$. 
So it is very much necessary to restrict $A$ to be within $\Omega(S_0)$ in \eqref{equ:Hrip}. With such a restriction, we can make $\delta$ to be close to 0. 

Let $T$ be a subset of $\{1, 2, \cdots, p\}$. For a vector $v$, $v_T$ is the restriction of $v$ on the subset $T$. For a matrix $A$ with $p$ columns, $A_T$ is the submatrix of $A$ consisting of only $|T|$ columns index by $T$.


 \begin{lemma}\label{lem:delta2}
 Assume $\|H\|\|H^\dagger\|=1$. Let $d_c(S_0)=d$ and $\ker(H)=\spn(v_1, \cdots, v_K)\in\R^p$ with $\|v_k\|_2=1$, then  $H$ has scaled-$S_0$-$\delta$-RINP \eqref{equ:Grip} with $\delta\leq \max_{|T|\leq d} \sum_{k=1}^K \|v_{k, T}\|_1\|v_{k}\|_\infty$.
 \end{lemma}
 \begin{proof}
 Let $v_k = (v_{1k}, v_{2k}, \cdots, v_{pk})$.
 $P=P_{\ker(H)}=\sum_{i=1}^Kv_iv_i^T$ and $P_{ij}=\sum_{k=1}^Kv_{ik}v_{jk}$.
 Since the spectral infinity norm is the maximum absolute row sum, for any index set $|T|\leq k$.
 \begin{align*}
 &\|P_T\|_\mi=\max_{i\in[1:p]}\sum_{j\in T}|P_{ij}|\leq \max_{i\in[1:p]}\sum_{j\in T}\sum_{k=1}^K|v_{ik}||v_{jk}|=\max_{i\in[1:p]}\sum_{k=1}^K|v_{ik}|\sum_{j\in T}|v_{jk}|\\
 &=\max_{i\in[1:p]}\sum_{k=1}^K|v_{ik}|\|v_{k, T}\|_1=\sum_{k=1}^K \|v_{k, T}\|_1\max_{i\in[1:p]}|v_{ik}|=\sum_{k=1}^K \|v_{k, T}\|_1\|v_{k}\|_\infty.
 \end{align*}
 The proof is finished using Lemma \ref{lem:delta}.
 \end{proof}
 
 We can see that the constant $\delta$ will be closely related to the null space property (see ~\cite{CDD09, CWW14}) of $H$ given each $v_k\in\ker(H)$. We will pursue this direction in future work.
 
 \begin{example}\label{exa:codim1}
 Let $\|H\|\|H^\dagger\|=1$ and $\ker(H)=\spn(v)$ be a 1-dimensional subspace of $\R^p$ and $d_c(S_0)=d$.  By Lemma \ref{lem:delta2}, $H$ has scaled RINP with 
   \begin{equation}\label{equ:deltav}
  \delta=\max_{|T|\leq d}  \|v_{T}\|_1\|v\|_\infty
  \end{equation}
  The right hand side of \eqref{equ:deltav} should be on the order of $d/p$ as long as the energy in coordinates of $v$ are somewhat equally distributed. The best case scenario will be when $|v_j|=1/\sqrt{p}$ in which case $\delta=\frac{d}{p}$. 
 \end{example}

\subsubsection{$H$ is Gaussian}
Recall that $H$ has $m$ rows and $p$ columns.  
Suppose $H$ is a random Gaussian matrix with i.i.d. entries drawn from $N(0, 1/m)$. The $1/m$ scaling is a common choice in the compressed sensing literature and is adopted here for convenience. However, this scaling is not essential, as the diagonal matrix $D$ can automatically adjust it. 

Combined with Lemma \ref{lem:delta}, the following lemma shows such $H$ has RINP with high probability.
\begin{lemma}
Suppose $H$ is an $m\times p$ random Gaussian matrix with i.i.d. entries drawn from $N(0, 1/m)$. Then
\begin{equation}\label{equ:Hgaussian}
 \Pr \left(\|(I-H^TH)x\|_{\infty} \leq \delta \|x\|_{\infty}  \right) \geq 1- \varepsilon, \quad \delta = C s\sqrt{ \frac{ \log(p/\varepsilon) }{ m } }.
\end{equation}
\end{lemma}
\begin{proof}
\[
 H^\top H = \frac{1}{m} \sum_{k=1}^m \bar h_k \bar h_k^\top.
\]
where $\bar h_k$ is the $k$th row of $H$ scaled up by $\sqrt m$, i.e., $\bar h_k = \sqrt{m}H[i,:]$.

We want to bound the quantity:
\[
\sup_{\substack{x \in \mathbb{R}^p \\ \|x\|_\infty \leq 1,\, \|x\|_0 \leq s}} \left\| (I - H^\top H) x \right\|_\infty.
\]

\paragraph{Entrywise bound.} For any \( i, j \in \{1, \dots, p\} \), consider the individual entries of \( \hat{\Sigma} :=H^TH\). We have
\[
\hat{\Sigma}_{ij} = \frac{1}{m} \sum_{k=1}^m \bar h_k^{(i)} \bar h_k^{(j)}.
\]

Since the entries of \( \bar h_k \) are sub-Gaussian, their products \( \bar h_k^{(i)} \bar h_k^{(j)} \) are sub-exponential random variables. Let us define:
\[
X_k := \bar h_k^{(i)} \bar h_k^{(j)} - \mathbb{E}[\bar h_k^{(i)} \bar h_k^{(j)}].
\]

Then \( \{X_k\}_{k=1}^m \) are independent, mean-zero, sub-exponential random variables. By Bernstein's inequality for sub-exponential random variables, we have:
\[
\Pr\left( \left| \frac{1}{m} \sum_{k=1}^m X_k \right| \geq \varepsilon \right)
\leq 2 \exp\left( -c m \min\left( \frac{\varepsilon^2}{K^2}, \frac{\varepsilon}{K} \right) \right),
\]
for some constant \( c > 0 \) where $K$ is  the sub-exponential norm of $X_k$ ($K\sim O(1)$). 

In particular, for small \( \varepsilon \lesssim K \), we have the simplified sub-Gaussian-style tail:
\begin{equation}\label{equ:tail}
\Pr\left( \left| \hat{\Sigma}_{ij} - \delta_{ij} \right| \geq \varepsilon \right)
\leq 2 \exp\left( - c m \varepsilon^2 \right),
\end{equation}
where $\delta_{ij} =  \mathbb{E}[\bar h_k^{(i)} \bar h_k^{(j)}]$ is the
Kronecker delta.

This bound \eqref{equ:tail} holds uniformly for each fixed pair \( (i, j) \). To obtain a uniform bound over all \( i, j \in [p] \), apply a union bound:
\[
\Pr\left( \max_{i,j} \left| \hat{\Sigma}_{ij} - \delta_{ij} \right| \geq \varepsilon \right)
\leq 2p^2 \exp\left( - c m \varepsilon^2 \right).
\]

Therefore, setting $\bar \varepsilon :=2p^2 \exp\left( - c m \varepsilon^2 \right) $, with  probability at least \( 1 - \bar \varepsilon \), we have:
\[
\max_{i,j} \left| \hat{\Sigma}_{ij} - \delta_{ij} \right| \leq C \sqrt{ \frac{ \log(p/\bar \varepsilon) }{ m } },
\]
for some constant \( C \) depending on the sub-Gaussian norm.

\paragraph{Action on sparse vectors.} For an \( s \)-sparse vector \( x \in \mathbb{R}^p \) with \( \|x\|_\infty \leq 1 \), for each $i$, the \( i \)-th coordinate of \( (I - H^\top H)x \) is:
\[
\left[ (I - H^\top H)x \right]_i = x_i - \sum_{j=1}^p \hat\Sigma_{ij}  x_j = \sum_{j \in \text{supp}(x)} x_j ( \delta_{ij} - \hat\Sigma_{ij} ).
\]

So, for all $i \in 1,...,p$,
\[
 \left| \left[ (I - H^\top H)x \right]_i \right| \leq \sum_{j \in \text{supp}(x)} \left| \delta_{ij} - \hat\Sigma_{ij}  \right| \leq s \cdot \max_{i,j} \left| \delta_{ij} - \langle \bar h_i, \bar h_j \rangle \right| \leq C s\sqrt{ \frac{ \log(p/\bar \varepsilon) }{ m } }.
\]
with probability at least $1-\bar\varepsilon$.


\end{proof}

\section{Conditions on low rank and sparse matrices}
\subsection{Incoherence condition}
We will analyze the incoherence condition \eqref{equ:incoherence} in this section. 
\begin{align}\label{equ:muHe1}
\mu_G(S)=\max_{A\in\Omega(S), \|A\|_\infty\leq1}\|GA\|\leq \|G\|\mu(S).
\end{align}

Recall that $d(S)$ is the maximum number of non-zero entries per row or per column of $S$. It is proven in~\cite[Proposition 3]{C11} that $\mu(S)\leq d(S)$, so we simply have 
\begin{equation}\label{equ:muHe}
\mu_G(S)\leq  \|G\| d(S).
\end{equation}

Let $\alpha(G):= \max_i\|g_i\|_2$ be defined as the maximum column norm of $G$, $\beta(U) := \max_i \|P_U e_i\|_2$ be the incoherence of the subspace $U$, and $\beta(U,G):= \max_i\|P_U g_i\|_2$ be the incoherence between columns of $G$ and the subspace $U$. Then we show $\xi_G(L)$ can be bounded in terms of the three quantities $\beta(U,G)$, $\alpha(G)$ and $\beta(V)$, where $L = U\Sigma V$ is its SVD.
\begin{align}\notag
\xi_G(L)& = \max_{B\in \Tc(L), \|B\|\leq1}\|G^T B\|_\infty  =\max_{B\in \Tc(L), \|B\|\leq1}\|G^T P_{\Tc(L)}B\|_\infty\leq \max_{ \|B\|\leq1}\|G^T P_{\Tc(L)}B\|_\infty \\\notag
&\stackrel{\eqref{equ:pt}}{=}\max_{ \|B\|\leq1}\|G^T(P_UB+(I-P_U)BP_V)\|_\infty\\\notag
& \leq \max_{B \textrm{ unitary} } \|G^T P_U B||_{\infty} + \max_{B \textrm{ unitary} } \|G^T(I-P_U)BP_V\|_{\infty} \\\notag
& =\max_{B \textrm{ unitary} } \max_{i,j} |g_i^T P_U Be_j| +  \max_{B \textrm{ unitary} } \max_{i,j} |g_i^T(I-P_U)BP_Ve_j| \\\notag
& \leq \max_{B \textrm{ unitary} }\max_i \|P_U g_i\|_2 \max_j \|B e_j\|_2 + \max_i \|(I-P_U)g_i\|_2\max_j \|BP_Ve_j\|_2\\\notag
&\leq \max_i \|P_U g_i\|_2+ \max_i \|g_i\|_2\max_j \|P_Ve_j\|_2\\\label{equ:xiHe}
& = \beta(U,G) + \alpha(G)\beta(V) \equiv \inc(L,G).
\end{align} 


Combining \eqref{equ:muHe} and \eqref{equ:xiHe}, Theorem \ref{thm:main} implies the following.
\begin{corollary}
Given $M_0=HS_0+L_0$ where $G=HD$ satisfies the $S_0$-$\delta$-RINP \eqref{equ:Hrip} with $\delta<1/3$ for some invertible diagonal matrix $D$.
If 
\begin{equation}\label{equ:ni}
d(S_0)\inc(L_0,G)<\frac{1-3\delta}{6\|G\|},
\end{equation}
then for any optimizer $(\hat S, \hat L)$ of \eqref{equ:p}, we must have $\hat S=S_0, \hat L=L_0$.  
\end{corollary}

 \begin{example}
 Following Example \ref{exa:codim1}, let $H$ be a square matrix whose nonzero singular values are all 1. This implies $\alpha(H)\leq1$ and $\|H\|=1$. We also let $\rank(H)=m-1$ and $\ker(H)=\spn(v)$ where each coordinate of $v$ has magnitude 1, then $H$ has \eqref{equ:Hrip} ($G=H$ here) with $\delta\leq d(S_0)/m$.  Let $d=d(S_0)$ for short. 
 \begin{align*}
 \left(\inc(L_0, H)+\frac{1}{2m}\right)d<\frac{1}{6}\Longleftrightarrow d\cdot\inc(L_0, H)<\frac{1-3d/m}{6}\Longrightarrow \eqref{equ:ni}.
 \end{align*}
 Furthermore, $\left(\inc(L_0, H)+\frac{1}{2m}\right)d<\frac{1}{6}$ is true if 
 $$ \left(\beta(U,H) + \beta(V)+\frac{1}{2m}\right)d<\frac{1}{6}.$$
  \end{example}

\subsection{Random sparse $S$ and low rank $L$ with a Gaussian $H$}
For the standard setting with an incoherent low-rank matrix 
$L$, a sparse matrix 
$S$, and Gaussian mask
$H$, we apply the theorem to derive an exact recovery guarantee in terms of the sparsity level and rank. 

A matrix $S$ is said to follow the \textbf{random sparsity model} with support size $s$ if its support is chosen uniformly at random from the collection of all support sets of cardinality $s$. There is no assumption made on the values of $S$ at its support. The random sparsity model is common and can be found in \cite{C11} for example.

A matrix $L\in\R^{m\times n}$ is said to follow the \textbf{right-side random orthogonal model} with rank $r$ if $L=U\Sigma V^T$ where the right singular vectors $V$ are drawn uniformly at random from the collection of rank $r$ partial isometries in $\R^{n\times k}$ and the left singular vectors $U$ are arbitrary. No restriction is placed on the singular values. $L$ is said to follow the \textbf{random orthogonal model} if both $U$ and $V$ are drawn uniformly at random from the collection of partial isometries.

We now estimate $\inc(L,H)$ where 
$L$ is drawn from the right-side random orthogonal model, and $H$ is a Gaussian measurement matrix.

Let $L= USV^T$, since $V$ is drawn randomly from the collection of rank $r$ partial isometries we have, by \cite{CR09},
\begin{equation}
\Pr\left(\max_{i}\|V^Te_i\|_2\leq\sqrt{\frac{r+16\sqrt{r\log n}}{n}}\right)\geq 1-2n^{-3}.
\end{equation}

Since $H \in \mathbb{R}^{m\times p}$ has i.i.d. entries drawn from $N(0, \frac{1}{m})$, 
its $i$th column  $h_i$ follows $\|h_i\|\sim \sqrt{\frac{1}{m}\chi_m^2}$, and by the union bound, we have
\[
 \Pr \left(\max_{i}\|h_i\|_2 \leq 1+\sqrt{\frac{2}{m}\log \frac{p}{\varepsilon}}\right) \geq 1-\varepsilon
\]

For the fix $U \in \mathbb{R}^{m\times r}$, using the fact that $Uh_i$ is still i.i.d. Gaussian vector, we have
\begin{align*}
& \Pr\left(\max_{i}\|U^Th_i\|_2 \leq  \sqrt{\frac{5r\log p/\varepsilon  }{m}}\right)\geq 1- \varepsilon.
\end{align*} 

Combining the above, we conclude that, with high probability \begin{equation}\inc(L,H)\leq \beta(U,H) + \alpha(H)\beta(V) =  \sqrt{\frac{5r\log p/\varepsilon  }{m}} + \left(1+\sqrt{\frac{2}{m}\log \frac{p}{\varepsilon}}\right)\sqrt{\frac{r+16\sqrt{r\log n}}{n}}.\label{equ:inc} 
\end{equation}

For the random sparsity model, we have the following lemma, whose proof can be found in the Appendix.
\begin{lemma}
Suppose $S_0\in\R^{p\times n}$ is drawn from the random sparsity model with support size $s$, then 
\begin{equation}\label{equ:S0r}
d(S_0)\leq\max\{\frac{s}{n}\log(n), \frac{s}{p}\log(p)\}, 
\end{equation}
with high probability. 
\end{lemma}

Assuming $m=n=p$,
\eqref{equ:ni} is satisfied with high probability if 
\begin{equation}
\frac{s}{n}\log^2(n)\sqrt{\frac{r+16\sqrt{r\log n}}{n}}\leq\frac{1-3\delta}{12}, \delta \lesssim \frac{s \sqrt{\log n}}{n^{3/2}}
\end{equation}
This means $s$ can be as large as $n^{3/2}/r$

\section{Numerical Experiments}\label{sec:num}
There are many algorithms for solving the robust PCA problem. The ADMM~\cite{boyd2011distributed} algorithm is a popular choice \cite{wang2013solving}. A non-convex algorithm based on alternating projections is presented in \cite{netrapalli2014non}, and an accelerated version that significantly improves the accuracy is proposed in \cite{CCW19}. The work  \cite{Ce20} proposes another non-convex algorithm using CUR decomposition.

Since our novel problem \eqref{equ:p} is convex and falls into the form of (3.1) in \cite{boyd2011distributed}, we use the standard 3-step ADMM algorithm \cite[(3.5)-(3.7)]{boyd2011distributed} to solve it. The general CVX package~\cite{grant2014cvx} was also tried, but it resulted slow and poor performance.
All numerical experiments were conducted on a MacBook Pro with Apple M3 chip and 8GB RAM, using Matlab 2023b. 

\subsection{Blurring Mask}
In the first experiment we let $m=n=p=100$. The sparse matrix $S_0$ is drawn from the random sparsity model with varying support size $s\in\{0.01mn, 0.03mn, 0.06mn, \cdots, 0.3mn\}$. Each nonzero entries of  $S_0$ are independently drawn from a uniform distribution in $[1, 2]$. The low rank matrix $L_0$ is drawn from the random orthogonal model with varying rank $r\in\{1, 4, 7, \cdots, 28\}$. 
The mask $H\in\R^{100\times100}$ is a rank 99 matrix with 
$$\ker(H)=\spn\{(1, -1, 1, -1, \cdots, 1, -1)\}.$$ 
Specifically, we first construct circulant $\tilde H$ whose first row is $[1, 1, 0, \cdots, 0]$, then $\tilde H=U_H\Sigma_HV_H^T$ is its SVD and we let $H=U_HV_H^T$. 
So $H$ has $S_0$-RINP constant $\delta=d_c(S_0)/100$ as argued in Example \ref{exa:codim1}.

For each fixed $s$ and $r$, we generate 8 instances of $L_0, S_0$ to produce $M_0=HS_0+L_0$. We then use ADMM to solve \eqref{equ:p} with $\gamma=1/\sqrt{m}=0.1$. The relative error of recovering $S_0$ and $L_0$ is defined as $\frac{\|\hat S-S_0\|_F}{\|S_0\|_F}$ and $\frac{\|\hat L-L_0\|_F}{\|L_0\|_F}$ respectively where $\|\cdot\|_F$ is the Frobenius norm. The relative errors are then averaged over these 8 trials and shown in Figure \ref{fig:1} as two heatmaps. Left of Figure \ref{fig:1} shows the relative error in recovering the sparse matrix and right of Figure \ref{fig:1} shows the  relative error in recovering the low rank matrix. The $x$-axis shows the sparsity level in terms of the proportion of the nonzero entries in $S_0$. The $y$-axis shows the rank in terms of ratio of the possible max rank. The color of each pixel indicates the averaged relative error at each corresponding sparsity and rank level. The performance of the recovery is quite well and surpasses the theoretical guarantee as expected. An interesting phenomenon is that the performance on recovering $S_0$ is a lot better than that of $L_0$.

\begin{figure}[htb]
\includegraphics[width=1\textwidth]{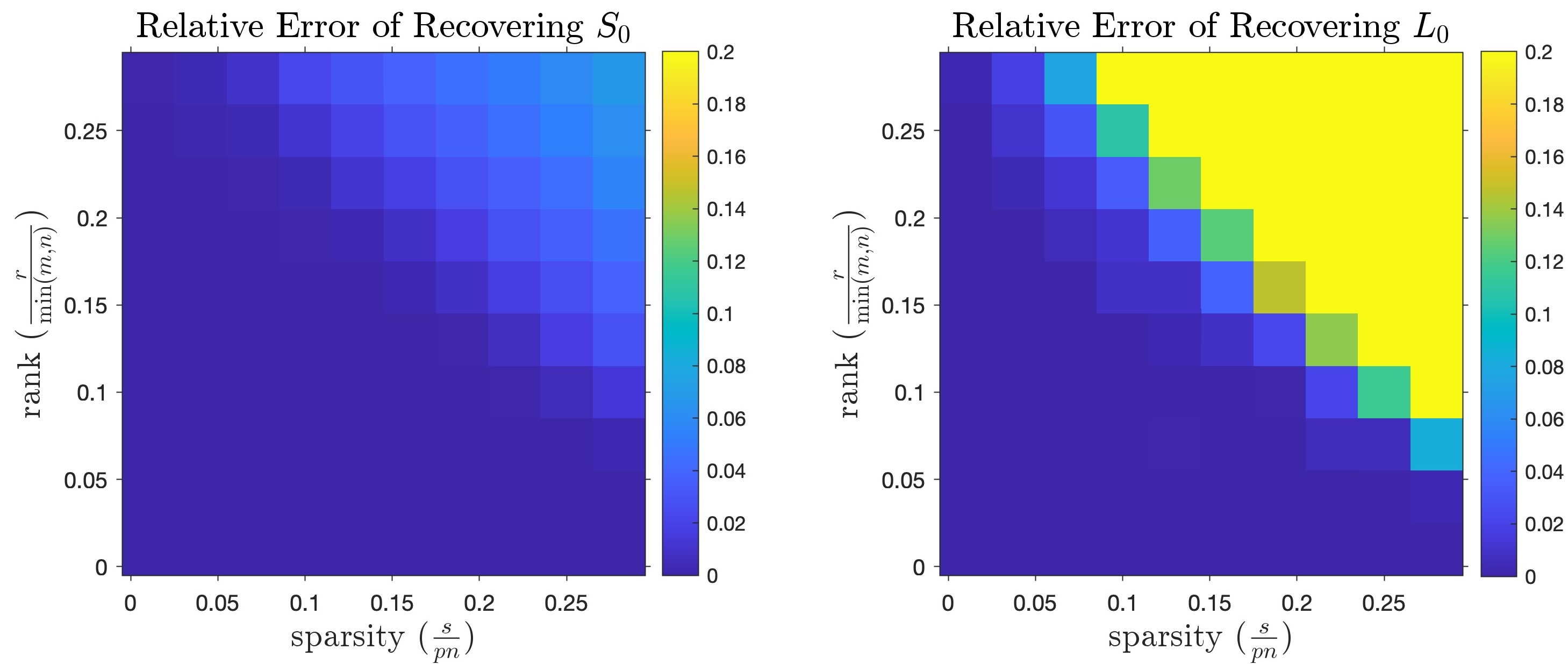}
\caption{Relative errors when $H$ is a blurring mask}\label{fig:1}
\end{figure}

\subsection{Random mask}
This experiment is similar to the design of the previous experiment, where the only difference is $H$ is random Gaussian with i.i.d. entries drawn from $N(0, 1/m)$. 

We also let $m=n=p=100$. The sparse matrix $S_0$ is drawn from the random sparsity model with varying support size $s\in\{0.01mn, 0.03mn, 0.06mn, \cdots, 0.3mn\}$. Each nonzero entries of  $S_0$ are independently drawn from a uniform distribution in $[1, 2]$. The low rank matrix $L_0$ is drawn from the random orthogonal model with varying rank $r\in\{1, 4, 7, \cdots, 28\}$. 
For each fixed $s$ and $r$, we generate 8 instances of $L_0, S_0$ to produce $M_0=HS_0+L_0$. We then use ADMM to solve \eqref{equ:p} with $\gamma=1/\sqrt{m}=0.1$. The relative error of recovering $S_0$ and $L_0$ is defined as $\frac{\|\hat S-S_0\|_F}{\|S_0\|_F}$ and $\frac{\|\hat L-L_0\|_F}{\|L_0\|_F}$ respectively where $\|\cdot\|_F$ is the Frobenius norm. The relative errors are then averaged over these 8 trials and shown in Figure \ref{fig:heat2} as two heatmaps. Left of Figure \ref{fig:1} shows the relative error in recovering the sparse matrix and right of Figure \ref{fig:1} shows the  relative error in recovering the low rank matrix. The $x$-axis shows the sparsity level in terms of the proportion of the nonzero entries in $S_0$. The $y$-axis shows the rank in terms of ratio of the possible max rank. The color of each pixel indicates the averaged relative error at each corresponding sparsity and rank level. 
We can see that the performance is not as good as in Figure \ref{fig:1}. Moreover, when support size exceeds 25\% of the entries, recovery is not possible no matter what the rank is.

\begin{figure}[htb]
\includegraphics[width=1\textwidth]{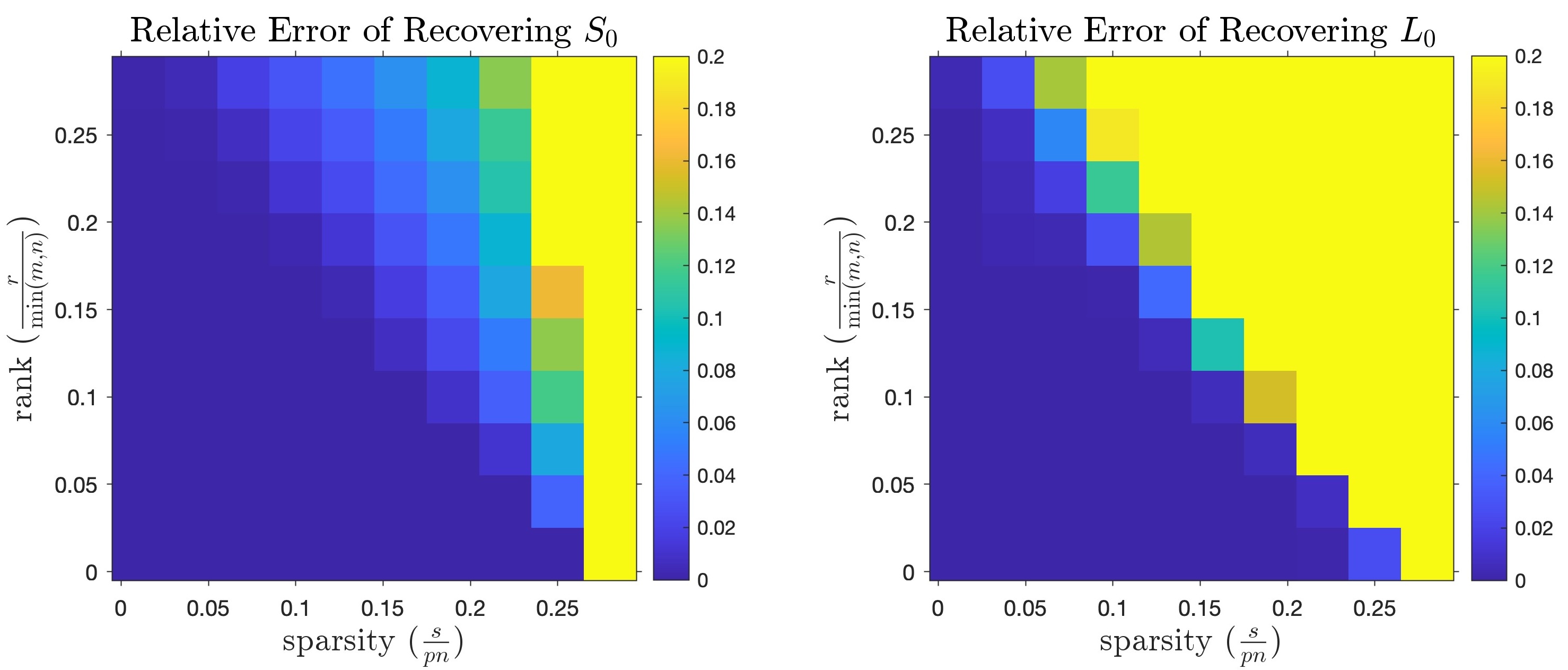}
\caption{Relative errors when $H$ is random}\label{fig:heat2}
\end{figure}
\subsection{EDA decomposition}\label{sec:num_eda}

We perform a simulated EDA decomposition following the signal model in the introduction. We have $m=240, p=160, n=50$.
The kernel $h\in\R^{160}$ is obtained by sampling the function
\begin{equation}\label{equ:filter}
f(t)=2(e^{-t/\tau_1}-e^{-t/\tau_2})
\end{equation}
at the rate of 4 samples per second in the interval $t\in[0, 40]$. This generates the discrete vector $(f_1, f_2, \cdots, f_{160})$. We use $\tau_1=2, \tau_2=0.75$. This is supported by psychophysiology literature such as~\cite{Ae05}. We will use a centered convolution, i.e., the corresponding $H\in\R^{240\times160}$ is the center block of
$\begin{bmatrix}
f_1&&&\\
\vdots&f_1&&\\
\vdots&\vdots&\ddots&\\
f_{160}&\vdots&&f_1\\
&f_{160}&&\vdots\\
&&\ddots&\vdots\\
&&&f_{160}
\end{bmatrix}$. 
We generate  $X=[x_1, \cdots, x_{n}]$ from the random sparsity model with the number of nonzero entries being $s$. Each column of $X$ represents an SCR event. We run our experiment with different values of $s$ from the set $\{4n, 6n, 8n, \cdots, 30n\}$. This means the number of SCR events in each EDA signal ranges in $\{4, 6, 8, \cdots, 30\}$, which is the $x$-axis of Figure \ref{fig:eda}. For the tonic component $T$, we first create a slow varying vector in $\R^{mn}$. It is then reshaped to an $m\times n$ matrix to make $T$. The matrix $T$ is very close to a rank-1 matrix. The noise matrix $E$ is produced by iid Gaussian with mean 0 and standard deviation 0.01.

For each fixed $s$, we generate 8 instances of $X, T, E$ as described above. Let $Y = T+HX+E$. We then solve $X$ and $T$ using $$
(\hat X, \hat T)=\argmin_{S, L}\gamma\|S\|_1+\|L\|_*, \quad\text{subject to }L+HS=Y.
$$
with $\gamma=1/\sqrt{n}$.

Figure \ref{fig:eda} shows the  relative recovery errors  averaged over 8 trials for each sparsity value. 

\begin{figure}[hbt]
\includegraphics[width=1\textwidth]{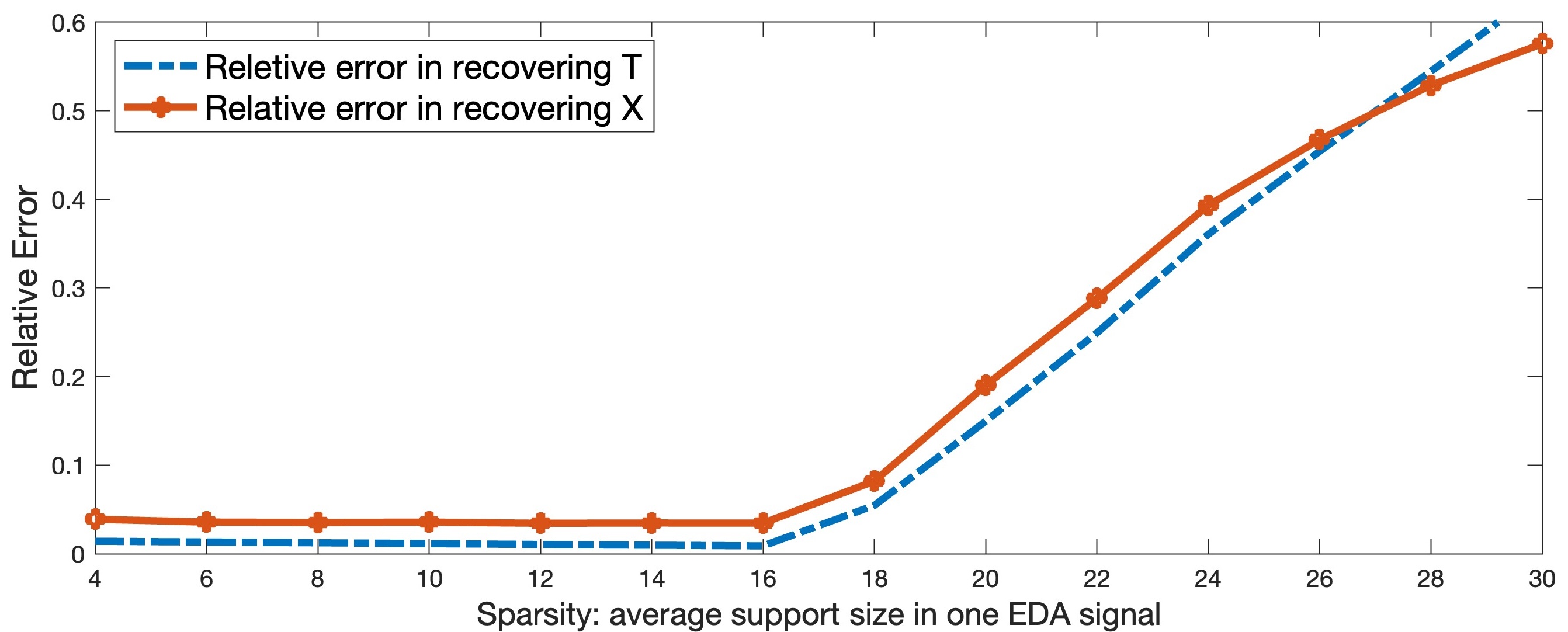}
\caption{Relative error in the EDA experiment as sparsity varies}\label{fig:eda}
\end{figure}

We observe from Figure \ref{fig:eda} that recovery of $X$ (and $T$) becomes poor when the number of SCR events or number of nonzero coordinates in one EDA signal ($\in\R^{160}$) is over 20.

In this EDA experiment, we see that program \eqref{equ:p} is robust with respect to the added noise $E$. This is not surprising at all and we will leave its theoretical guarantee as future work.

%

\section*{Acknowledgements}
X. Chen is partially funded by NSF DMS-2307827. X. Chen would like to thank UNC Wilmington's Research Hub initiative. R. Wang is funded by NSF CCF-2212065.
\section*{Appendix}
%

Let $X$ be a hypergeometric distribution denoted by $X\sim \text{HYP}(N, K, s)$, where $N$ is the population size, $K$ is the number of success states in the population, and $s$ is the number of draws. Let $p=K/N$. It is well known that $\E(X) = sp$. Moreover, $X$ approaches the Bernoulli distribution (with $s$ draws and success rate $p$) as the population grows, but is more concentrated around its mean than the Bernoulli distribution.  We have~\cite{GW17, S74}
\begin{equation}\label{equ:hyp}
\Pr(X\geq \E(X) + \sqrt{s}\lambda)\leq\exp(-\frac{2\lambda^2}{1-\frac{s-1}{N}}), \quad\text{ for any }\lambda >0.
\end{equation}

\begin{lemma}
Suppose $S\in\R^{m\times n}$ is drawn randomly from all collection of matrices with a fixed support size $s$. Let $r = \frac{s}{mn}$ be the proportion of the nonzero entries. Let $d_r(S), d_c(S)$ be the maximum number of nonzero entries per row, per column respectively, then
\begin{align}
\Pr(d_r(S)\geq\frac{s}{m}\log m)\leq m^{-\frac{nr}{m(2-r)}},\\
\Pr(d_c(S)\geq\frac{s}{n}\log n)\leq n^{-\frac{mr}{n(2-r)}}.
\end{align}
\end{lemma}
\begin{proof}
Let $X_i$ be the number of nonzero entries in the $i$th row of $S$. By computing its probability mass function, it is easy to see that $X_i\sim\text{HYP}(mn, n, s)$. We have $\E(X_i) = s/m$. Note that $X_i$'s are not independent from each other. Let $\rho = 1-\frac{s-1}{mn}=1-r+\frac{1}{mn}$, then
\begin{align}\label{equ:r1}
\Pr(X_i\geq \frac{s}{m}\log m)=\Pr(X_i-\E(X_i)\geq\frac{s}{m}\log m - \frac{s}{m})\stackrel{\eqref{equ:hyp}}{\leq}
\exp\left(-\frac{2s}{m^2\rho}\left(\log m-1\right)^2\right).
\end{align}
\begin{align*}
&\Pr(d_r(S)\geq\frac{s}{m}\log m) = \Pr\left(\cup_{i=1}^m(X_i\geq\frac{s}{m}\log m)\right)\leq m\Pr(X_1\geq \frac{s}{m}\log m)\\
\stackrel{\eqref{equ:r1}}{\leq}&m\exp\left(-\frac{2s}{m^2\rho}\left(\log m-1\right)^2\right)\leq m^{-\frac{nr}{m\rho}}\leq m^{-\frac{nr}{m(2-r)}}.
\end{align*}
The ratio $r$ and the dimension ratio $n/m$ are considered somewhat fixed, so $\frac{s}{m^2\rho}=\frac{nr}{mp}$ is a bounded constant.
 The last inequality above comes from the crude estimate when $m$ is big enough
$$\log\left(m\exp\left(-\frac{2s}{m^2\rho}\left(\log m-1\right)^2\right)\right)=\log(m)-\frac{2s}{m^2\rho}\left(\log m-1\right)^2\leq-\frac{nr}{m\rho}\log m.$$

\end{proof}
%

\bibliographystyle{amsplain}
\bibliography{ref_matrix_sep}

\providecommand{\bysame}{\leavevmode\hbox to3em{\hrulefill}\thinspace}
\providecommand{\MR}{\relax\ifhmode\unskip\space\fi MR }
\providecommand{\MRhref}[2]{%
  \href{http://www.ams.org/mathscinet-getitem?mr=#1}{#2}
}
\providecommand{\href}[2]{#2}
\begin{thebibliography}{10}

\bibitem{Ae05}
David~M Alexander, Chris Trengove, P~Johnston, Tim Cooper, JP~August, and Evian
  Gordon, \emph{Separating individual skin conductance responses in a short
  interstimulus-interval paradigm}, Journal of neuroscience methods
  \textbf{146} (2005), no.~1, 116--123.

\bibitem{boyd2011distributed}
Stephen Boyd, Neal Parikh, Eric Chu, Borja Peleato, Jonathan Eckstein, et~al.,
  \emph{Distributed optimization and statistical learning via the alternating
  direction method of multipliers}, Foundations and Trends{\textregistered} in
  Machine learning \textbf{3} (2011), no.~1, 1--122.

\bibitem{edaguide}
Jason~J Braithwaite, Derrick~G Watson, Robert Jones, and Mickey Rowe, \emph{A
  guide for analysing electrodermal activity (eda) \& skin conductance
  responses (scrs) for psychological experiments}, Psychophysiology \textbf{49}
  (2013), no.~1, 1017--1034.

\bibitem{CCW19}
HanQin Cai, Jian-Feng Cai, and Ke~Wei, \emph{Accelerated alternating
  projections for robust principal component analysis}, Journal of Machine
  Learning Research \textbf{20} (2019), no.~20, 1--33.

\bibitem{Ce20}
HanQin Cai, Keaton Hamm, Longxiu Huang, Jiaqi Li, and Tao Wang, \emph{Rapid
  robust principal component analysis: Cur accelerated inexact low rank
  estimation}, IEEE Signal Processing Letters \textbf{28} (2020), 116--120.

\bibitem{CR09}
E.J. Cand{\`e}s and B.~Recht, \emph{Exact matrix completion via convex
  optimization}, Found Comput Math \textbf{9} (2009), 717--772.

\bibitem{C08}
Emmanuel~J Candes, \emph{The restricted isometry property and its implications
  for compressed sensing}, Comptes rendus. Mathematique \textbf{346} (2008),
  no.~9-10, 589--592.

\bibitem{CLMW11}
Emmanuel~J Cand{\`e}s, Xiaodong Li, Yi~Ma, and John Wright, \emph{Robust
  principal component analysis?}, Journal of the ACM (JACM) \textbf{58} (2011),
  no.~3, 1--37.

\bibitem{C11}
Venkat Chandrasekaran, Sujay Sanghavi, Pablo~A Parrilo, and Alan~S Willsky,
  \emph{Rank-sparsity incoherence for matrix decomposition}, SIAM Journal on
  Optimization \textbf{21} (2011), no.~2, 572--596.

\bibitem{chaspari2014sparse}
Theodora Chaspari, Andreas Tsiartas, Leah~I Stein, Sharon~A Cermak, and
  Shrikanth~S Narayanan, \emph{Sparse representation of electrodermal activity
  with knowledge-driven dictionaries}, IEEE Transactions on Biomedical
  Engineering \textbf{62} (2014), no.~3, 960--971.

\bibitem{CWW14}
Xuemei Chen, Haichao Wang, and Rongrong Wang, \emph{A null space analysis of
  the $\ell_1$-synthesis method in dictionary-based compressed sensing},
  Applied and Computational Harmonic Analysis \textbf{37} (2014), no.~3,
  492--515.

\bibitem{CDD09}
A.~Cohen, W.~Dahmen, and R.~DeVore, \emph{Compressed sensing and best k-term
  approximation}, Journal of the American mathematical society \textbf{22}
  (2009), 211--231.

\bibitem{grant2014cvx}
Michael Grant and Stephen Boyd, \emph{Cvx: Matlab software for disciplined
  convex programming, version 2.1}, 2014.

\bibitem{GW17}
Evan Greene and Jon~A Wellner, \emph{Exponential bounds for the hypergeometric
  distribution}, Bernoulli: official journal of the Bernoulli Society for
  Mathematical Statistics and Probability \textbf{23} (2017), no.~3, 1911.

\bibitem{HKZ11}
Daniel Hsu, Sham~M Kakade, and Tong Zhang, \emph{Robust matrix decomposition
  with sparse corruptions}, IEEE Transactions on Information Theory \textbf{57}
  (2011), no.~11, 7221--7234.

\bibitem{edacs}
Swayambhoo Jain, Urvashi Oswal, Kevin~Shuai Xu, Brian Eriksson, and Jarvis
  Haupt, \emph{A compressed sensing based decomposition of electrodermal
  activity signals}, IEEE Transactions on biomedical engineering \textbf{64}
  (2016), no.~9, 2142--2151.

\bibitem{netrapalli2014non}
Praneeth Netrapalli, Niranjan UN, Sujay Sanghavi, Animashree Anandkumar, and
  Prateek Jain, \emph{Non-convex robust pca}, Advances in neural information
  processing systems \textbf{27} (2014).

\bibitem{R70}
R.~T. Rockafellar, \emph{Convex analysis}, Princeton University Press, 1970.

\bibitem{S74}
Robert~J Serfling, \emph{Probability inequalities for the sum in sampling
  without replacement}, The Annals of Statistics (1974), 39--48.

\bibitem{wang2013solving}
Xiangfeng Wang, Mingyi Hong, Shiqian Ma, and Zhi-Quan Luo, \emph{Solving
  multiple-block separable convex minimization problems using two-block
  alternating direction method of multipliers}, Pacific Journal of Optimization
  \textbf{11} (2015), 645--667.

\bibitem{W92}
G~Alistair Watson, \emph{Characterization of the subdifferential of some matrix
  norms}, Linear Algebra Appl \textbf{170} (1992), no.~1, 33--45.

\bibitem{YGL15}
Ming Yin, Junbin Gao, and Zhouchen Lin, \emph{Laplacian regularized low-rank
  representation and its applications}, IEEE transactions on pattern analysis
  and machine intelligence \textbf{38} (2015), no.~3, 504--517.

\end{thebibliography}

%
%
%
%
\end{document}